\documentclass[a4paper,11pt]{article}

\usepackage[a4paper,margin=1in]{geometry}
\usepackage{amsmath,amsthm,mathtools,thmtools,thm-restate}
\usepackage{amsfonts}
\usepackage{graphicx}
\usepackage{lmodern}
\usepackage[utf8]{inputenc}
\usepackage[T1]{fontenc}
\usepackage{booktabs,color,doi,graphicx,latexsym,url,xcolor,xspace}
\usepackage{microtype}
\usepackage{enumitem}
\usepackage{framed}
\usepackage{hyperref}
\usepackage[capitalize]{cleveref}
\usepackage{pgfplotstable}
\usepackage{pgfplots}
\usepackage{fullpage}
\usepackage{orcidlink}

\pgfplotsset{compat=1.17}

\hypersetup{
	colorlinks=true,
	linkcolor=black,
	citecolor=black,
	filecolor=black,
	urlcolor=black,
}

\newtheorem{theorem}{Theorem}

\newtheorem{lemma}[theorem]{Lemma}

\theoremstyle{definition}
\newtheorem{definition}[theorem]{Definition}
\theoremstyle{remark}

\DeclareMathOperator{\median}{median}

\newcommand{\fakeparagraph}[2]{\par\noindent\textbf{#1}\hspace{1em}#2}

\newenvironment{myabstract}
{\list{}{\listparindent 1.5em%
		\itemindent    \listparindent
		\leftmargin    1cm
		\rightmargin   1cm
		\parsep        0pt}%
	\item\relax}
{\endlist}

\newenvironment{mycover}
{\list{}{\listparindent 0pt
		\itemindent    \listparindent
		\leftmargin    1cm
		\rightmargin   1cm
		\parsep        0pt}%
	\raggedright
	\item\relax}
{\endlist}

\newcommand{\myemail}[1]{\,$\cdot$\, {\small #1}}
\newcommand{\myaff}[1]{\,$\cdot$\, {\small #1}\par\medskip}

\begin{document}
	
\begin{mycover}
{\huge\bfseries\boldmath Fast Dynamic Programming in Trees in the MPC Model\par}
\bigskip
\bigskip

\textbf{Chetan Gupta}
\orcidlink{0000-0002-0727-160X} 
\myemail{chetan.gupta@aalto.fi}
\myaff{Aalto University, Finland}

\textbf{Rustam Latypov}
\orcidlink{0000-0001-7124-3067}
\myemail{rustam.latypov@aalto.fi}
\myaff{Aalto University, Finland}

\textbf{Yannic Maus}
\orcidlink{0000-0003-4062-6991}
\myemail{yannic.maus@ist.tugraz.at}
\myaff{TU Graz, Austria}

\textbf{Shreyas Pai}
\orcidlink{0000-0003-2409-7807}
\myemail{shreyas.pai@aalto.fi}
\myaff{Aalto University, Finland}

\textbf{Simo Särkkä}
\orcidlink{0000-0002-7031-9354}
\myemail{simo.sarkka@aalto.fi}
\myaff{Aalto University, Finland}

\textbf{Jan Studen\'y}
\orcidlink{0000-0002-9887-5192}
\myemail{jan.studeny@aalto.fi}
\myaff{Aalto University, Finland}

\textbf{Jukka Suomela}
\orcidlink{0000-0001-6117-8089}
\myemail{jukka.suomela@aalto.fi}
\myaff{Aalto University, Finland}

\textbf{Jara Uitto}
\orcidlink{0000-0002-5179-5056}
\myemail{jara.uitto@aalto.fi}
\myaff{Aalto University, Finland}

\textbf{Hossein Vahidi}
\orcidlink{0000-0002-0040-1213}
\myemail{hossein.vahidi@aalto.fi}
\myaff{Aalto University, Finland} 

\bigskip
\end{mycover}

\begin{myabstract}
	\fakeparagraph{Abstract.}
    We present a deterministic algorithm for solving a wide range of \emph{dynamic programming problems} in trees in $O(\log D)$ rounds in the massively parallel computation model (MPC), with $O(n^\delta)$ words of local memory per machine, for any given constant $0 < \delta < 1$.
    Here $D$ is the diameter of the tree and $n$ is the number of nodes---we emphasize that our running time is independent of $n$.
    
    Our algorithm can solve many classical \emph{graph optimization problems} such as maximum weight independent set, maximum weight matching, minimum weight dominating set, and minimum weight vertex cover.
    It can also be used to solve many \emph{accumulation} tasks in which some aggregate information is propagated upwards or downwards in the tree---this includes, for example, computing the sum, minimum, or maximum of the input labels in each subtree, as well as many inference tasks commonly solved with belief propagation.
    Our algorithm can also solve any \emph{locally checkable labeling problem} (LCLs) in trees.
    Our algorithm works for any reasonable representation of the input tree; for example, the tree can be represented as a list of edges or as a string with nested parentheses or tags.
    The running time of $O(\log D)$ rounds is also known to be necessary, assuming the widely-believed $2$-cycle conjecture.

    Our algorithm strictly improves on two prior algorithms:
    \begin{enumerate}
        \item Bateni, Behnezhad, Derakhshan, Hajiaghayi, and Mirrokni [ICALP'18] solve problems of these flavors in $O(\log n)$ rounds, while our algorithm is much faster in low-diameter trees. Furthermore, their algorithm also uses randomness, while our algorithm is deterministic.
        \item Balliu, Latypov, Maus, Olivetti, and Uitto [SODA'23] solve only locally checkable labeling problems in $O(\log D)$ rounds, while our algorithm can be applied to a much broader family of problems.
    \end{enumerate}
\end{myabstract}

\thispagestyle{empty}
\newpage

\section{Introduction}\label{sec:intro}

In this work we present a general, unified algorithm framework for solving a very wide variety of computational problems related to \emph{tree-structured data} in a massively parallel setting. Some examples of tasks that can be solved with our algorithm include:
\begin{itemize}
	\item Solving traditional graph optimization problems in trees (e.g., finding a maximum-weight independent set or mini\-mum-weight dominating set).
	\item Solving constraint-satisfaction problems in trees (e.g., finding a solution to any locally checkable labeling problem \cite{Naor1995}, as well as many generalizations of the theme).
	\item Analyzing large text documents with tree-structured data (e.g., processing large XML \cite{XML} documents).
	\item Aggregating information in trees (e.g., calculating the sum of inputs in each subtree \cite{GibbonsCS94}---this is a generalization of the classical prefix sum operation \cite{Ladner1980} from directed paths to rooted trees).
	\item Performing statistical inference in tree-structured graphical models (e.g., computations that are in the classical sequential setting commonly done with belief propagation \cite{Koller_book09}).
\end{itemize}

\subsection{Setting: MPC Model}

We work in the usual \emph{massively parallel computation} model (MPC) \cite{Karloff2010}. The size of the input is $n$ words---here $n$ is much larger than what fits in the local memory of a single computer, and therefore the input is distributed among multiple computers. The local memory of each computer is $\Theta(n^\delta)$ words, for some constant $0 < \delta < 1$. We have got $\Theta(n^{1-\delta})$ computers that take part in the computation, and hence in total $\Theta(n)$ words of distributed memory.

We will assume that the key bottleneck is communication between computers, and hence the time complexity is measured in the \emph{number of communication rounds}. We will assume that in one round each computer can send up to $\Theta(n^\delta)$ words to other computers and receive up to $\Theta(n^\delta)$ words from other computers. In essence, you can send everything you have in your local memory to someone else, and you can receive whatever fits in your local memory. When we refer to the \emph{running time} in this work, we always refer to the number of communication rounds (but we point out already here that in our algorithms local computation will also be lightweight).

\subsection{Prior Work: Solving LCL Problems Fast}\label{ssec:intro-prior-lcl}

In a recent work, Balliu, Latypov, Maus, Olivetti, and Uitto \cite{alkida22connectivity} presented efficient MPC algorithms for finding connected components, rooting trees, and solving so-called \emph{locally checkable labeling problems} (LCLs) in forests. As we directly build on their work, we will first briefly discuss their contributions.

LCL problems were first formalized by \cite{Naor1995}. These are graph problems that can be specified by listing a \emph{finite} set of feasible local neighborhoods. For example, ``$5$-coloring a graph of maximum degree $4$'' is an example of an LCL problem; we can list all properly $5$-colored neighborhoods that may occur in a graph of maximum degree $4$. Typically, constraint satisfaction problems are LCLs (as long as we have bounded degrees and a finite label set), while global optimization problems like maximum-weight independent set are not LCLs. 

The algorithms in \cite{alkida22connectivity} run in $O(\log D)$ rounds, where $D$ is the diameter of the input graph, with no asymptotic global memory overhead. Finding connected components and rooting are their main contributions, but here we are primarily interested in the part that solves LCL problems.

The algorithm for solving LCL problems consists of phases that compress the input graph; there are $O(1)$ phases and each phase takes $O(\log D)$ rounds. After phase $i$, they define a new LCL problem on the compressed graph such that its solution can be expanded into a solution for the LCL problem defined on the graph of phase $i-1$. After performing $O(1)$ phases the graph is compressed into a single node (the root of the tree) for which any LCL problem is trivially solved. The algorithm then finishes off with $O(1)$ reversal phases that decompress all compressed parts while simultaneously spreading the correct LCL solution to the decompressed parts of the graph.

\subsection{Key New Contributions: Unified Framework for Dynamic Programming Problems}\label{ssec:intro-new}

We build on \cite{alkida22connectivity} and present a new algorithm framework, with the following main features:
\begin{enumerate}
	\item We are able to solve a \emph{much} broader family of problems in $O(\log D)$ time---instead of solving only LCL problems, we can solve a much more general family of so-called \emph{dynamic programming problems} (see \cref{def:dp_problems}). We refer to \cref{tab:problem-examples} for some examples of the applicability of our framework in comparison with \cite{alkida22connectivity}.
	\item The prior algorithm \cite{alkida22connectivity} intermixes the tasks of compressing the tree and constructing the solution for an LCL. We show that it is possible to separate the concerns, as we will outline in \cref{ssec:intro-3step}. In particular, we can first use $O(\log D)$ rounds to construct a hierarchical clustering of the graph, and then with the help of the clustering, we can solve any dynamic programming problem in $O(1)$ rounds.
\end{enumerate}
The fastest prior algorithm for dynamic programming in the MPC model was the algorithm by Bateni, Behnezhad, Derakhshan, Hajiaghayi, and Mirrokni \cite{BateniBDHM18icalp,hajiaghayi18tree_dp}, but the running time of their algorithm is $O(\log n)$, which can be much worse than $O(\log D)$ in low-diameter trees, and moreover their algorithm is randomized while our algorithm is deterministic.

\begin{table}
    \newcommand{\yes}{$\checkmark$}
    \newcommand{\no}{---}
    \centering
    \begin{tabular}{l c c}
    \toprule
    Problem & Prior work \cite{alkida22connectivity} & This work \\
    \midrule
    Vertex coloring                                 & \yes & \yes \\
    Edge coloring                                   & \yes & \yes \\
    Maximal independent set                         & \yes & \yes \\
    \midrule
    Maximum weight independent set                  & \no  & \yes \\
    Maximum weight matching                         & \no  & \yes \\
    Minimum weight dominating set                   & \no  & \yes \\
    Minimum weight vertex cover                     & \no  & \yes \\ 
    Weighted max-SAT problem                        & \no  & \yes \\
    Longest path problem                            & \no  & \yes \\
    Sum coloring problem                            & \no  & \yes \\
    Counting matchings modulo $k$                   & \no  & \yes \\
    \midrule
    Tree median problem                             & \no  & \yes \\
    Inference in Bayesian graphical models          & \no  & \yes \\  
    Evaluating arithmetic expressions               & \no  & \yes \\
    Verifying the structure of XML-like documents   & \no  & \yes \\
    Computing the sum, minimum, or maximum \\
    \quad of the input labels in each subtree       & \no & \yes \\
    \bottomrule
    \end{tabular}
    \caption{Examples of problems solved with our framework and the prior work \cite{alkida22connectivity}.}\label{tab:problem-examples}
\end{table}

\subsection{Simple Three-Step Approach}\label{ssec:intro-3step}

Our algorithm framework proceeds in three steps:
\begin{enumerate}
	\item\label{step:intro-1} We turn the input into a \textbf{standard representation}; the running time of this phase is $O(\log D)$ rounds.
	We work with tree-structured data, but such data can be represented in different forms: we might have e.g.\ an unrooted tree that is represented as a long list of undirected edges, or we might have a rooted tree that is represented as a very long string (e.g.\ a string with nested parentheses or nested pairs of opening and closing tags). We will turn any such representation into a more convenient standard form: we will have a \emph{rooted tree} that is represented as \emph{list of directed edges}.
	We show that for a wide range of commonly-used representations of tree-structured data, this can be solved in $O(\log D)$ rounds. This is the only step  that depends on the precise input representation. We will give the details in \cref{sec:input_rep}.
	\item\label{step:intro-2} We construct a \textbf{hierarchical clustering} of the tree; the running time of this phase is $O(\log D)$ rounds.
	We will introduce the properties of the hierarchical clustering in \cref{ssec:intro-decomp}. We will show that such a clustering can be computed in $O(\log D)$ rounds. This step is fully generic---it depends neither on the input representation nor on the problem that we are solving. We will give the details in \cref{sec:decomposition}.
	\item\label{step:intro-3} We \textbf{solve the problem of interest}; the running time of this phase is $O(1)$ rounds.
	We show that we can solve a very wide variety of problems related to tree-structured data in $O(1)$ rounds, given the hierarchical clustering. We will give the details in \cref{sec:base}.
\end{enumerate}
Overall, this approach makes it possible to solve various computational problems in $O(\log D)$ rounds in trees. Furthermore, this results in algorithms that are \emph{conditionally optimal}: many problems that can be solved with this framework require $\Omega(\log D)$ rounds, assuming the (widely-believed) two-cycle conjecture~\cite{Andoni2018, Roughgarden18, Assadi2019, Ghaffari2019}.
The conjecture states that $\Omega(\log n)$ MPC-rounds are required to decide whether an input graph consists of a cycle of length $n$ or two cycles of length $n/2$, even if a polynomial number of machines is available.
It is known that this conjecture implies that finding connected components requires $\Omega(\log D)$ rounds~\cite{Behnezhad2019, Coy-det2022}, which in turn can be used to show that solving a subset of dynamic programming problems on trees requires $\Omega(\log D)$ rounds~\cite{alkida22connectivity}.

The main conceptual message of our work is this:
\begin{framed}
	\noindent
	There exists a single, convenient, universal representation that one can use as a starting point for designing very efficient massively parallel algorithms for tree-structured data.
\end{framed}
\noindent
We emphasize that the hierarchical clustering needs to be computed only once for a given input topology, and it can be reused for any dynamic programming problem and any input values.

\subsection{Hierarchical Clustering}\label{ssec:intro-decomp}

\begin{figure}
	\centering
	\includegraphics[page=2]{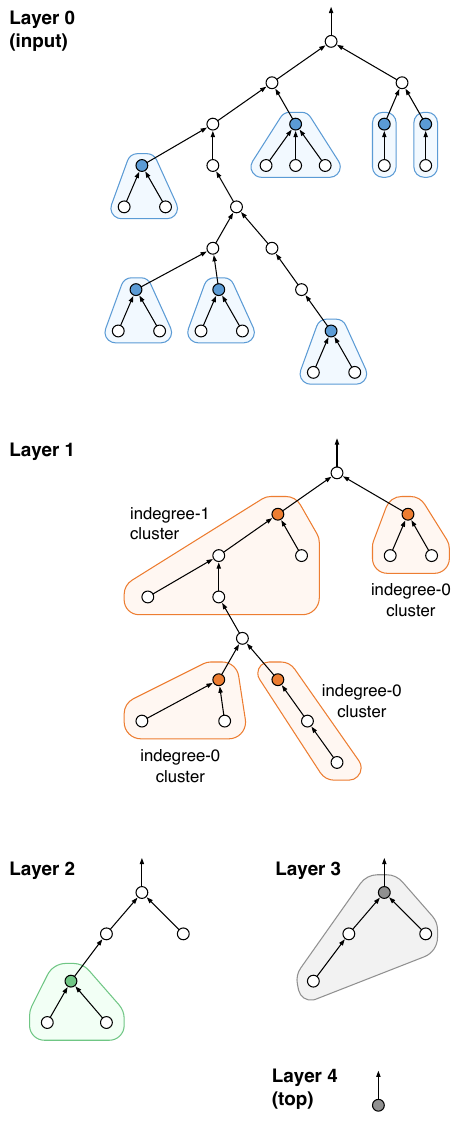}
	\caption{Our hierarchical clustering consists of constantly many layers. Layer $0$ is the input tree. At each layer we compress some disjoint collection of clusters so that eventually we have got only one node left. Each cluster contains at most $n^\delta$ nodes, each cluster has got exactly one outgoing edge, and there are zero or one incoming edges.}\label{fig:decomp}
\end{figure}

Our hierarchical clustering is illustrated in \cref{fig:decomp}. For convenience, we assume that all nodes of the tree have outdegree $1$; to ensure this we add at the root an additional virtual edge pointing outside the tree---this edge will be ignored when solving the problem of interest.

To construct the hierarchical clustering, we start with the original tree (this is our layer~$0$). To obtain layer $i+1$, we contract a \emph{cluster} of nodes into one node. The key properties that we ensure are:
\begin{itemize}[noitemsep]
	\item Each cluster contains only $O(n^\delta)$ nodes.
	\item Each cluster has outdegree $1$.
	\item Each cluster has indegree $0$ or $1$.
	\item There are only $O(1)$ layers, and the topmost layer consists of only one cluster.
\end{itemize}
We formally define the hierarchical clustering in \cref{sec:decomposition}, and we further show that it not only exists, but can also be computed in $O(\log D)$ rounds in the MPC model.

\subsection{Dynamic Programming Problems}\label{ssec:intro-what-how}

Our main focus is on problems that we will call \emph{dynamic programming problems}; as we will see in \cref{sssec:intro-example}, it is straightforward to adapt many typical optimization problems into this framework:
\begin{definition}\label{def:dp_problems}
	A dynamic programming problem (DP problem) is a computational problem in trees with the following properties:
	\begin{enumerate}
		\item The task is to compute a \emph{label} for each edge.
		\item We can summarize each cluster $C$ with a \emph{dynamic programming table} $f(C)$ that can be represented with $O(1)$ words.
		\item Given such summaries for all nodes that form a cluster $C$, we can compute in the dynamic programming table $f(C)$, using only $O(|C|)$ words of additional space; see \cref{fig:bottom-to-top}.
		\item We can compute the label for the outgoing edge of the top-level cluster $C$ given $f(C)$.
		\item Assuming that we know the labels of the incoming and outgoing edges of a cluster $C$ and the dynamic programming tables for each component of $C$, we can also compute the labels of all internal edges of cluster $C$, using only $O(|C|)$ words of additional space; see \cref{fig:top-to-bottom}.
	\end{enumerate}
\end{definition}
\noindent Here the labels of the edges are an abstraction of whatever is the specific task we are solving, while the dynamic programming tables are auxiliary data structures needed during the algorithm.

\begin{figure}
	\centering
	\includegraphics[page=5]{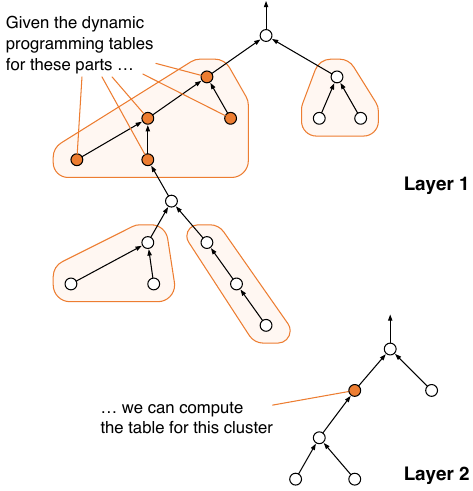}
	\caption{From bottom to top: given the summaries inside a cluster, we assume we can compute the summary for the entire cluster.}\label{fig:bottom-to-top}
\end{figure}

\begin{figure}
	\centering
	\includegraphics[page=6]{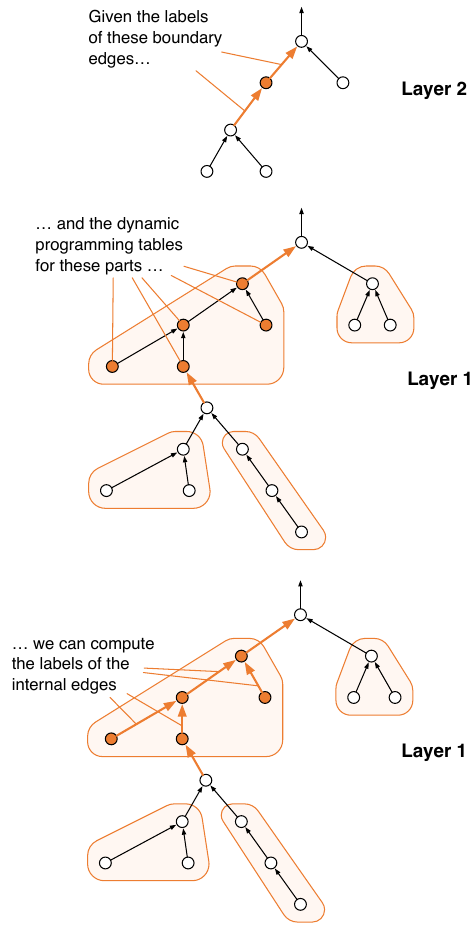}
	\caption{From top to bottom: given the solutions at the boundary edges, we assume we can compute the solution also for the internal edges.}\label{fig:top-to-bottom}
\end{figure}

\subsubsection{Example: Maximum-Weight Independent Set}\label{sssec:intro-example}

We will use the maximum-weight independent set problem (MaxIS) as a running example: in our input, each node has a nonnegative weight, and the task is to find a maximum-weight subset of nodes $X \subseteq V$ such that there is no edge $(u,v) \in E$ with $u \in X$ and $v \in X$.

Now the MaxIS problem is an example of a DP problem, with the following interpretation:
\begin{itemize}
	\item The label of the edge $(u,v)$ indicates whether $u \in X$.
	\item Let $C$ be an indegree-$0$ cluster, where $(u,v)$ is the outgoing edge. Then $f(C)$ is a table with two elements: (1)~the weight of the heaviest independent set in $C$ such that $u \in X$, and the (2)~the weight of the heaviest independent set in $C$ such that $u \notin X$.
	\item Let $C$ be an indegree-$1$ cluster, where $(u,v)$ is the outgoing edge and $(s,t)$ is the incoming edge. Now $f(C)$ is a table with four elements: the weight of the heaviest independent set in $C$ for all combinations of $u \in X$ vs.\ $u \notin X$ and $t \in X$ vs.\ $t \notin X$.
\end{itemize}
It is now easy to work out the details of the bottom-up and top-down phases. Note that the way we handle indegree-$0$ clusters is, in essence, identical to the classical centralized, sequential algorithm that solves MaxIS in trees (see e.g.\ \cite[Sect.\ 6.7]{Dasgupta2008}). The way we handle indegree-$1$ clusters can be seen as a special case of the centralized, sequential algorithm that solves MaxIS in bounded-treewidth graphs \cite{Bodlaender88}: we can summarize clusters with a constant number of interfaces to the rest of the graph, and we can merge such clusters.

\subsubsection{Beyond Dynamic Programming}

While we use the term \emph{dynamic programming} here to capture the problem family of interest, we would like to emphasize that there is a broad range of problems that are compatible with this framework even if one does not usually think that they have got anything to do with dynamic programming (recall \cref{tab:problem-examples}).

\subsection{Technicality: Very High Degrees}\label{ssec:intro-high-deg}

So far we have ignored one technical difficulty: what if our input tree has nodes of degree more than $n^\delta$. In such a case it is impossible to find small clusters, as the cluster that contains node $v$ will also contain all of its children.

Fortunately, for many problems such as MaxIS, we can easily modify the input and the problem slightly, so that we replace each node $v$ of degree more than $n^{\delta/2}$ with an $O(1)$-depth tree $T_v$. The new edges are equipped with additional labels so that we can handle them correctly in the dynamic programming algorithm and ensure that all nodes in $T_v$ make the same consistent choice.

We discuss this in more detail in \cref{ssec:high-degree,ssec:high-degree-dp}. To summarize, we can solve in any DP problem (\cref{def:dp_problems}), as long as we have degrees at most $n^{\delta/2}$ or we can reduce the degree as needed by replacing high-degree nodes with low-degree trees.

\subsection{Further Discussion on Related Work}

\paragraph{Bateni, Behnezhad, Derakhshan, Hajiaghayi, and Mirrokni.}

The prior work \cite{BateniBDHM18icalp,hajiaghayi18tree_dp} presents an MPC algorithm for dynamic programming in trees in $O(\log n)$ rounds in the MPC model. While the precise family of problems that they handle is phrased somewhat differently, the spirit is the same---they can also solve problems similar to the MaxIS problem.

Our work strictly improves on their work in two ways: our running time is $O(\log D)$, which is conditionally optimal, while their running time is $O(\log n)$, and our algorithm is deterministic, while their algorithm uses randomness.

In \cref{subsec:tree_median}, we also show how to solve a problem called \emph{tree median} using our framework. This is a problem engineered so that it does \emph{not} satisfy the property of \emph{binary adaptability}, which is a technical requirement used in \cite{BateniBDHM18icalp,hajiaghayi18tree_dp}. Informally, in binary adaptable problems one can replace high-degree nodes with binary trees, and hence it is sufficient to solve dynamic programming problems in bounded-degree trees; however, the tree median problem does not admit such a straightforward degree reduction. We hope this problem serves as a demonstration of the broad applicability of our framework, also beyond what was considered in prior work.

\paragraph{Balliu, Latypov, Maus, Olivetti, and Uitto.}

The prior work \cite{alkida22connectivity} presents an MPC algorithm for solving locally checkable labeling problems (LCLs) in trees in $O(\log D)$ rounds in the MPC model. Our running time is the same, but we solve a much broader family of problems (recall \cref{tab:problem-examples}).

We make use of many subroutines and ideas developed in \cite{alkida22connectivity}. For example, we make use of their algorithm for rooting a tree, and the idea of the hierarchical clustering as well as its key properties are due to them.

From the conceptual perspective, the key difference is that their work presents a single (arguably rather complicated) algorithm that intermixes the tasks of clustering the tree and constructing the solution for an LCL. The hierarchical clustering is rather implicit, and it has got properties that make it not directly applicable for solving a broad variety of problems: for example, arbitrarily long paths are compressed into one cluster, which will then no longer fit in the memory of one computer, and leaf nodes are aggressively eliminated, which is not compatible with all dynamic programming problems. In our algorithm the hierarchical clustering is built first, explicitly, and our clustering has got convenient properties that allow us to do per-cluster computations locally inside one computer, and it also allows us to tackle a broad range of problems.

\paragraph{Other Related Work.}

While our technique is conditionally optimal for the \emph{family} of dynamic programming problems, there are many problems that allow faster algorithms in certain cases.
For example, Balliu, Brandt, Fischer, Latypov, Maus, Olivetti, and Uitto \cite{Balliu2022-exp} consider classes of LCL problems that are local in nature, such as the MIS problem.
For many classes of natural problems, they give MPC algorithms that are much more efficient than $\Theta(\log D)$ for high diameter graphs.

Im, Moseley, and Sun \cite{im17stoc_dp} consider dynamic programming in the MPC model for problems that are not directly related to tree-structured inputs.

There is a related yet more powerful model called AMPC in which machines, in addition to the regular MPC operations,  can perform a sublinear number of (adaptive) queries to a distributed hash table per round. In the AMPC model, the problem of computing subtree sizes can be solved in $O(1)$ rounds \cite{MPC-via-remote-access}. 

In the classic PRAM model, problems of the same flavor have been studied already in the 1990s---for example, Gibbons, Cai, and Skillicorn \cite{GibbonsCS94} present an algorithm for upwards and downwards accumulation in trees that runs in $O(\log n)$ time. We emphasize that while $\Omega(\log n)$ is a natural lower bound for all such problems in the PRAM model, we can nevertheless achieve a running time of $O(\log D)$ in the MPC model.


\section{Preliminaries}\label{sec:preliminaries}

We make use of the following primitives: \emph{sorting} an array of $n$ elements and computing \emph{prefix sums} in an array of $n$ elements. Both of these operations can be solved in the MPC model with a deterministic algorithm in $O(1)$ rounds, see \cite{Goodrich99,GoodrichSZ11,CzumajDP21}.

\section{Input Representations}\label{sec:input_rep}

Our algorithm in \cref{sec:decomposition,sec:base} will assume that the input tree is rooted and it is given as a set of directed edges such that each edge goes from a child to its parent node. However, in addition to this standard representation there are various other ways to represent a tree using an array. In this section, we define other commonly used representations and show that we can transform the input from any of these representations to an array of directed edges in $O(1)$ round. In \cref{subsec:from_standard}, we will show how our algorithm framework makes it possible to turn the standard representation back to any of these representations.

\subsection{Definitions}\label{ssec:repr-def}
We consider tree-structured data represented in one of the following forms; we use the tree $T$ illustrated in \cref{fig:tree} as an example:
\begin{itemize}
    \item \textbf{List-of-edges:} This is the representation that our algorithm works with. Each element in the input array contains
    a pair of integers that represents a directed edge in a tree going from a child to its parent. Tree $T$ can be described as an array $[(1,4), (2,3), (5,4), (4,3)]$, if we use the labeling of the nodes given in \cref{fig:tree}.
    \item \textbf{String-of-parentheses:} In this representation, the tree is given as an array of properly nested parentheses or, equivalently, opening and closing tags. Each node in the tree is represented by two 
     parentheses ``('' and ``)''. We can interpret the array as a rooted tree in a bottom-up manner, with the leaf nodes represented as an empty pair of parentheses ``()''.
     The outermost pair of parentheses represents the root node. For example, $T$ can be represented as an array  $[{(},{(},{(},{)},{(},{)},{)},{(},{)},{)}]$.
    \item \textbf{BFS-traversal:} The array represents the BFS-traversal of the tree: the indices of the array denote the nodes in the tree in the BFS order, and an array element contains the index of the parent node. Tree $T$ can be represented as $[-,1,1,2,2]$.
    \item\textbf{DFS-traversal:} Similar to the above, the tree is given as an array that represents a DFS traversal of the tree. Tree $T$ can be represented as $[-,1,2,2,1]$.
    \item \textbf{Pointers-to-parents:} Similar to the above, but the nodes are ordered arbitrarily. Tree $T$ can be represented as $[4,3,-,3,4]$, if we order the nodes according to their labels in \cref{fig:tree}.
\end{itemize}

\begin{figure}
    \centering
      \includegraphics[page=7]{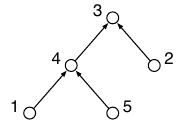}
      \caption{Tree $T$ used as an example in \cref{ssec:repr-def}.}
      \label{fig:tree}
\end{figure}

\subsection{Normalizing the Representation}

If the tree is originally given as a list of undirected edges, we can first root the tree at an arbitrary node and orient the edges in $O(\log D)$ rounds, using the algorithm from \cite{alkida22connectivity}.

BFS-traversal, DFS-traversal, and pointers-to-parents already represent the input as a set of directed edges in different manners, and hence it is easy to turn them into a list-of-edges representation. The nontrivial part
is to prove that we can obtain the \textbf{list-of-edges representation from string-of-parentheses} in $O(1)$ rounds in the MPC model.

We will first show how we can do this transformation for $\delta = 1/2$, i.e., assuming there are $m = \sqrt{n}$ computers each with $O(\sqrt{n})$ memory, and then we describe how to generalize the same strategy for any $\delta$.

Let $A$ be the array that contains properly nested parentheses. We assume that each opening parentheses ``('' in $A$ will represent a node in the tree.
Now for each open parenthesis, we need to find its parent open parenthesis.

Initially, $A$ is evenly distributed over $\sqrt{n}$ computers
$N_0,\ldots N_{m-1}$ such that $N_i$ contains the elements $A[i\sqrt{n}],\dotsc, A[(i+1)\sqrt{n} -1]$. Let $A[i]$ and $A[j]$ be 
two opening parentheses such that $i <j$. We know that $A[i]$ is the parent of $A[j]$ if all the parentheses from $A[i+1]$ to $A[j-1]$ are 
properly nested. If $A[i]$ is the parent of $A[j]$ and both of them are stored in the same computer, then the computer can easily identify
$A[i]$ as the parent of $A[j]$. The challenge is to identify the parent node if $A[i]$ and $A[j]$ are stored in different computers.

Notice that if $A[p]$ and $A[q]$ are a pair of opening and closing parentheses that denote the same node and both
are stored in some computer $N_i$ then $A[p]$ cannot be the parent of $A[k]$ if $A[k]$ is stored in some other computer. Thus, let us 
\emph{cancel out} properly nested pairs of parentheses stored in a single computer. Now the remaining parentheses inside each computer $N_i$, will be
nothing but a (possibly empty) sequence of closing parentheses followed by a (possibly empty) sequence of opening parentheses, for example, ``)))))(((''. Let $S_i$ be the array of remaining
parentheses in $N_i$.

Computer $N_i$ computes a pair $(c_i,o_i)$ where $c_i$ and $o_i$ is the number of closing and opening parentheses in $S_i$, and broadcasts it
to all the other computers. Using this information, for each node we can identify the array $S_i$ that contains its parent and also the index of 
the parent in $S_i$ as follows. For each open parentheses $A[j]$ stored at $N_i$, $N_i$ locally computes $l_j$ and $r_j$ that denote the number of closing and 
opening parentheses on the left and right side of $A[j]$, respectively, in $S_i$. Then $S[j]$ (stored in $N_a$) is the parent of $A[k]$ (stored in $N_b$) 
where $j<k$ and $a < b$, if $a$ is the largest integer such that \[r_j + \sum_{x=a+1}^b (o_x - c_x) - l_k = 0,\] which can be computed in $O(1)$ rounds by $N_b$.

To identify the index of the parent of a node in $A$, we need to do some more calculations. For each node $v$ we produce two tuples:
\begin{itemize}
    \item Type 1: $[i,j,1,v]$ denotes that node $v$ is stored at the $j$th index of $S_i$---this information is readily available for the computer that holds node $v$.
    \item Type 2: $[i,j,2,v]$ denotes that the \emph{parent} of node $v$ is stored at the $j$th index of $S_i$---this information can be computed as described above by the computer that holds node $v$.
\end{itemize}
This way we will have $n$ tuples in total in the system, and we can sort them in $O(1)$ rounds. Once sorted, in the array there will always be one tuple of type 1, representing a node $v$, followed by zero or more tuples of type 2, representing the children of $v$. This way we can identify all parent--child edges in $O(1)$ rounds.

\subsubsection{Low-memory Version}

Above, we assumed that we have got $\delta = 1/2$. Let us now see how the strategy generalizes to $\delta = 1/k$ for any $k$. In this case we will use a $k$-level strategy.

At level $\ell = 1, \dotsc, k$, we conceptually split the input in chunks of length $n^{\ell/k}$. Let $i$ be the parent of node $j$. We say that an edge $(i,j)$ is \emph{local} if $i$ and $j$ are in the same chunk, and otherwise it is global. We maintain the following invariant after level $\ell$:
\begin{itemize}
    \item We have already discovered all edges $(i,j)$ that are local.
    \item For each chunk we have computed a summary $(c, o)$ that denotes the number of closing and opening parentheses inside the chunk, after cancelling properly nested parentheses inside the chunk.
\end{itemize}
Assume $C$ is a chunk at level $\ell+1$ that consists of sub-chunks $C_1, \dotsc, C_c$ at level $\ell$; here by definition $c = n^{1/k}$. Now all computers that hold parts of chunk $C$ can learn the summaries $(c,o)$ for each sub-chunk $C_1, \dotsc, C_c$, as this information fits in their local memory. By following the same strategy as what we used in the case $\delta = 1/2$, we can now compute all local edges inside chunk $C$, as well as compute a summary $(c,o)$ for the entire chunk $C$. Hence, given the invariants at level $\ell$ we can in $O(1)$ time satisfy the invariants at level $\ell+1$.

If $\delta$ is not a convenient rational number $1/k$, we can round it down and let one computer with $O(n^\delta)$ memory play the role of many computers with $O(n^{1/k})$ memory each, and the above scheme applies.


\section{Hierarchical clustering} \label{sec:decomposition}
In this section we present an $O(\log D)$-round algorithm that computes the hierarchical clustering required for our dynamic programming algorithm (see \cref{sec:base}). Note that the clustering does not depend on the problem that we want to solve afterwards. 

\begin{figure*}
	\centering
	\includegraphics[page=3,scale=0.97]{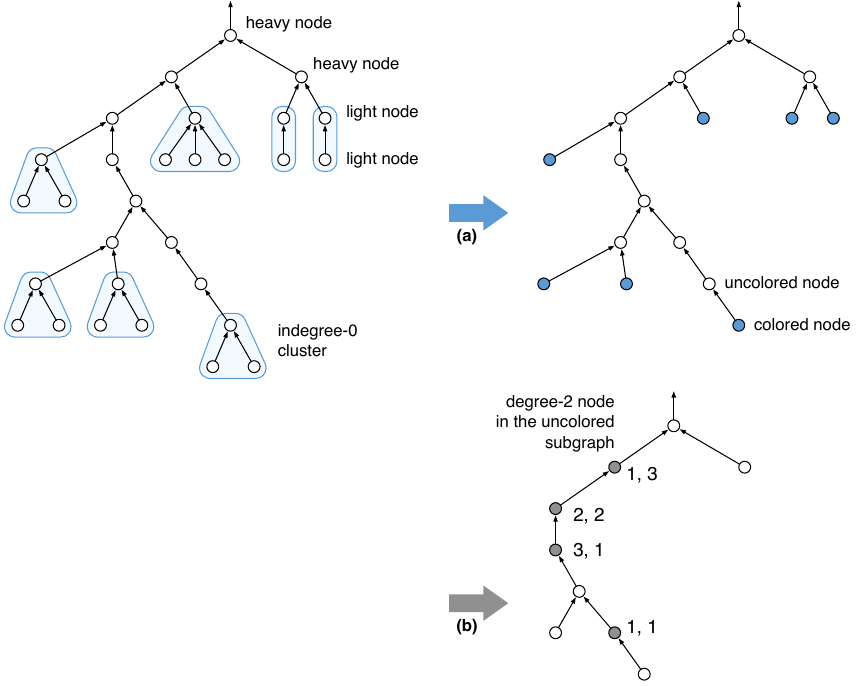}
	\caption{(a)~Creating indegree-zero clusters. (b)~Creating indegree-one clusters: we identify paths formed by degree-$2$ nodes in the subgraph induced by uncolored nodes and calculate their positions in the path both upwards and downwards.}\label{fig:create-clusters}
\end{figure*}

\subsection{Definitions}

We will now formalize the idea of hierarchical clustering that we introduced in \cref{ssec:intro-decomp}; see \cref{fig:decomp} for an illustration.

\begin{definition}[cluster]
	A cluster $C$ is a set such that each element is either a node $u_{i}$ or another cluster $C_{i}$. We recursively define the set of nodes that participate in $C$ as
  	\[
	V(C) = \bigcup_{C_{i} \in C} V(C_{i}) \cup \{u_{i} \mid u_{i} \in C\}.
	\]
	We require that the cluster $C$ contains at most $n^\delta$ elements, and the set of cut edges $(V(C), V \setminus V(C))\subseteq E$ has exactly one outgoing edge and at most one incoming edge.
\end{definition}

We classify clusters into two types based on the number of incoming edges: indegree-zero and indegree-one.

\begin{definition}[hierarchical clustering]\label{def:heirarchical-clustering}
	A \emph{hierarchical clustering} of a rooted tree $T=(V,E)$ is a collection of sets $S_{0},S_{1},\dots,S_{L}$ called layers such that $L=O(1)$ and the following are satisfied
	\begin{enumerate}
		\item each $S_i$ consists of nodes or clusters,
		\item $S_{0}=V$,
		\item For $i\ge 1$, (i) the nodes in $S_{i}$ are also nodes in $S_{i-1}$ and (ii) the clusters of $S_{i}$ form a partition of the remaining elements of $S_{i-1}$,
		\item $S_{L}$ contains one element which is a cluster.
	\end{enumerate}
\end{definition}
While it is easiest to grasp the clustering as a standalone \emph{graph-theoretic concept} in order to use it algorithmically, we need to assign cluster IDs and store certain pointers between a cluster and its nodes/clusters, etc.
More formally, we give each cluster $C \in S_{i}$ a unique cluster ID, and pointers to and from the clusters and nodes of $S_{i-1}$ that are contained in $C$. Since a cluster has exactly one outgoing and at most one incoming edge, we can contract each cluster in $S_{i}$ into a node, such that the resulting graph forms a tree $T_{i}$ where each edge corresponds to an edge of the original tree.

\subsection{Constructing the Clustering}

As discussed in \cref{sec:input_rep}, we can without loss of generality assume that the input is a rooted tree $T = (V, E)$ with $n$ nodes, represented as a list of edges. We will further assume that the maximum degree is $n^{\delta/2}$, but we will see how to overcome this limitation in \cref{ssec:high-degree}. By sorting the edges, we can also assume that each node and its incident edges are hosted on the same machine. Our goal is to construct a hierarchical clustering as in \cref{def:heirarchical-clustering}.

\subsubsection{High-Level Idea}
We will mostly follow the same ideas as what happens in the algorithm of \cite{alkida22connectivity}. However, there are two key differences that we will highlight in what follows, and we will also need to prove that the number of layers is still bounded by a constant.

We say that a subtree is a \emph{caterpillar} if it is a tree containing a central path and all other nodes are within distance $1$ from the path. We will alternate between two steps, for $O(1)$ iterations:
\begin{enumerate}
	\item Create indegree-zero clusters: we identify nodes $v$ such that we can replace the entire subtree $T(v)$ rooted at $v$ with a cluster.
	\item Create indegree-one clusters: we identify a disjoint set of caterpillars that we can replace with clusters.
\end{enumerate}

In \cite{alkida22connectivity}, they entirely removed what we call indegree-zero clusters, and then they only needed to contract long paths. Furthermore, they contracted arbitrarily long paths, while our clusters cannot be too large. Nevertheless, we can show that we make enough progress and we can finish after $O(1)$ pairs of such steps.

In our algorithm we will \emph{color} the nodes that correspond to indegree-zero clusters instead of removing them. Then we can largely follow the process and the analysis of \cite{alkida22connectivity} for the uncolored parts of the tree. As the colored nodes are always leaf nodes, and as each node can have at most $n^{\delta/2}$ neighbors, if we put into each cluster up to $n^{\delta/2}$ uncolored nodes, together with their colored neighbors the size of a cluster will be bounded by $n^{\delta}$, as needed.

\subsubsection{Creating Indegree-Zero Clusters}

Following \cite{alkida22connectivity}, we define that a node $v$ with more than $n^{\delta/2}$ uncolored nodes in its subtree $T(v)$ is called \emph{heavy}, and the rest of the nodes are \emph{light}.

We apply the following result from Lemma 6.13 of \cite{alkida22connectivity} to the uncolored subgraph (i.e., the subgraph induced by the uncolored nodes): there exists a deterministic optimal space $O(\log D)$-time MPC algorithm (CountSubtreeSizes) in which every node $v$ learns either the exact size of $T(v)$ or that $|T(v)| > n^{\delta/2}$.

With this information, we can identify each node $u$ such that $u$ is light but its parent $v$ is heavy. We apply Lemma 6.14 from \cite{alkida22connectivity}: there exists a deterministic optimal space $O(\log D)$-time MPC algorithm (GatherSubtrees) to collect $T(u)$ into the machine hosting $u$ for each such node $u$. Then, we replace $T(u)$ with an indegree-zero cluster, which is then represented as a colored node---see \cref{fig:create-clusters}. The overall running time is $O(\log D)$. The size of the cluster will be bounded by $n^{\delta}$, as there were only $n^{\delta/2}$ uncolored nodes, each with at most $n^{\delta/2}$ colored leaf nodes attached to it.

\subsubsection{Creating Indegree-One Clusters}

Now we are ready to describe the second step: creating indegree-one clusters. The idea is to identify long paths in the uncolored subgraph. A long path in the uncolored subgraph corresponds to a caterpillar if we also take into account the colored nodes.

We apply Lemma 6.17 from \cite{alkida22connectivity} to the uncolored subgraph: there exists a deterministic $O(\log D)$-time MPC algorithm (CountDistances) in which each degree-$2$ node knows its distance to both endpoints of the path formed by degree-$2$ nodes---see \cref{fig:create-clusters}.

Using the distances, we will split each path $P$ formed by degree-$2$ nodes in the uncolored subgraph into sub-paths of length at most $n^{\delta/2}$ (i.e., nodes with distance value $1, \dots, n^{\delta/2}$ form the first sub-path and so on). We call these sub-paths path fragment $P'$. We collect each fragment in a single machine and form a cluster $C$ by including also all colored nodes connected to $P'$. This will result in a caterpillar $C$, and as the maximum degree of the graph was $n^{\delta/2}$, the size of the cluster is at most $n^\delta$, as required. The overall running time of this step is $O(\log D)$.

\subsection{Number of Layers}

By construction, all clusters are sufficiently small. We still need to show that the number of layers is bounded by a constant:
\begin{lemma}\label{lem:decomp-layers}
	The number of layers in the hierarchical clustering we created is $O(1)$.
\end{lemma}

To prove \cref{lem:decomp-layers}, consider first an alternative process $\Pi_1$ where we delete indegree-zero clusters instead of marking them colored, and in which we replace arbitrarily long paths with one edge, similar to \cite{alkida22connectivity}. We can show:
\begin{lemma}
	Each iteration of process $\Pi_1$ makes the tree smaller by a factor of $\Omega(n^{\delta/2})$.
\end{lemma}
\begin{proof} 
	Say we start with a tree $T_{0}$ with $n_{0}$ nodes.
	Let there be $n_{1}$ nodes in the tree $T_{1}$ obtained after we delete the indegree-zero clusters and replace all paths with a single indegree-one cluster. This means that all paths are of length at most 1. Consider a tree $T_{1}'$, which is $T_{1}$ except all paths are replaces with an edge. Notice that $|T_{1}'| \ge n_{1}/2$, and $T_{1}$ has the same number of leaves as $T_{1}'$. Now, in $T_{1}'$ there are no nodes with degree $2$. And since any tree has at least as many leaves as nodes of degree $3$ or more, $T_{1}'$ has at least $|T_{1}'|/2$ leaves, which means that there are at least $n_{1}/4$ leaves in $T_{1}$.

	Consider a leaf node $v$. Since $v$ was not removed, it must have been heavy, and hence the subtree rooted at $v$ has size $ > n^{\delta/2}$. Hence, the number of nodes before we started our process clustered was $n_{0} \ge (n_{1}/4) \cdot n^{\delta/2}$.
	Therefore, the number of nodes in each clustering step falls by a factor of $n^{\delta/2}$.
\end{proof}

Then slightly modify the process; let $\Pi_2$ be a process in which we still delete indegree-zero clusters instead of marking them colored, but we replace arbitrarily long paths with one node and two edges.
\begin{lemma}
	Each iteration of process $\Pi_2$ makes the tree smaller by a factor of $\Omega(n^{\delta/2})$.
\end{lemma}
\begin{proof}
	In essence, $\Pi_2$ behaves as if we first performed one iteration of $\Pi_1$ and then subdivided some edges. The subdivision only increases the number of nodes by a factor of two.
\end{proof}

Finally, let $\Pi_3$ be a process in which we still delete indegree-zero clusters instead of marking them colored, but we replace long paths with a sequence of clusters, each with at most $n^{\delta/2}$, similar to our real process. We can show:
\begin{lemma}\label{lem:pi-3}
	$O(1)$ iterations of process $\Pi_3$ makes the tree smaller by a factor of $\Omega(n^{\delta/2})$.
\end{lemma}
\begin{proof}
	If we iterate $\Pi_3$ for more than $2/\delta$ iterations, each path gets contracted into a path with only one node. Hence, $2/\delta$ iterations of $\Pi_3$ makes at least as much progress as one iteration of $\Pi_2$.
\end{proof}

\cref{lem:decomp-layers} now follows by observing that $\Pi_3$ describes accurately what happens in the uncolored subgraph in our real process:

\begin{proof}[Proof of \cref{lem:decomp-layers}]
By applying \cref{lem:pi-3} iteratively for $O(1)$ times to the uncolored subgraph, we can see that the uncolored part gets contracted into one node, and at that point the entire graph will fit in one indegree-zero cluster.
\end{proof}

\subsection{Handling High-Degree Nodes}\label{ssec:high-degree}

So far we have assumed that the tree that is given as input has degree at most $n^{\delta/2}$. The general solution to overcome this limitation is to replace high-degree nodes with $O(1)$-depth subtrees.

Let us now briefly describe how to implement it in $O(1)$ rounds in the MPC model. We can sort the original list of edges by the parent node identifier. Now whenever a single machine holds more than $n^{\delta/2}$ edges with the same parent $u$, it introduces new nodes whose parent is $u$ and these new nodes become the new parent of $n^{\delta/2}$ children of $u$. We repeat this for $O(1)$ steps until all nodes have sufficiently low degrees. Throughout the process, we keep track of the \emph{type} of the edge: whether it is an \emph{original} edge or an \emph{auxiliary} edge created while splitting high-degree nodes---this information is needed then later when we solve the DP problem (see \cref{ssec:high-degree-dp}).

This process will increase the number of nodes and the diameter by only a constant factor. Hence, if we now apply the clustering algorithm, the running time is still $O(\log D)$ rounds, where $D$ is the diameter of the \emph{original} tree.


\section{Solving DP Problems}\label{sec:base}

Now we will show how we can use the hierarchical clustering computed in \cref{sec:decomposition} to solve dynamic programming problems (recall \cref{def:dp_problems}).

\subsection{From Bottom to Top}

Let $L = O(1)$ be the number of layers in the hierarchical clustering. We fill in the dynamic programming tables in $L$ iterations, by maintaining the following invariant:
\begin{definition}[bottom-up invariant]
	After iteration $i = 0, 1, \dotsc, L$, each cluster $C$ of layer $i$ is labeled with its dynamic programming table $f(C)$, and all other nodes are labeled with their original inputs.
\end{definition}
This invariant is trivial to satisfy in the beginning, as layer $0$ is our input tree and there are no clusters yet.

Now assume that we satisfy the invariant before iteration $i > 0$. Now each node that still participates in the computation knows both its cluster identifier for layer $i$ and either its input or its dynamic programming table. Furthermore, this information fits by assumption in $O(1)$ words. We can now sort the array of cluster identifiers and node labels and this way ensure that data related to one cluster is stored consecutively. Now one cluster spans at most two machines; with one additional routing step we can ensure that each cluster is fully contained inside one machine.

Now we can locally summarize each cluster $C$, by applying the sequential algorithm that we assumed exists. Finally, we have a summary $f(C)$ for each cluster. We can then apply sorting again to move the summary $f(C)$ back to the array location that we use to store information for cluster $C$. In essence, this enables us to solve the operation illustrated in \cref{fig:bottom-to-top} for each cluster in parallel.

Eventually, we have computed the dynamic programming tables for all clusters at all layers.

\subsection{From Top to Bottom}

Now we proceed to solve the problem, i.e., to fill in the labels of the edges. We proceed through the layers now in the reverse order, maintaining the following invariant:
\begin{definition}[top-down invariant]
	After iteration $i = L, L-1, \dotsc, 0$, we have computed the labels of all edges $(u,v)$ in the tree that corresponds to layer $i$, and this information is stored together with node $u$.
\end{definition}
This invariant can be satisfied for $i = L$: there is only one edge in the tree, the outgoing edge of the topmost cluster $C$, and by assumption given $f(C)$ we can label this edge.

Now assume we satisfy the invariant before iteration $i < L$. Now if $C$ is a cluster that appears in layer $i$, we can use sorting to ensure that the $C$ is aware of both the label of its outgoing edge and the label of its incoming edge (if any). Then we again to reorganize data so that the nodes of layer $i-1$ that form a cluster $C$ at layer $i$ are stored in the same computer. We can apply the sequential algorithm to now label all internal edges of $C$. In essence, this enables us to solve the operation illustrated in \cref{fig:top-to-bottom} for each cluster in parallel.

Eventually, we have computed the labels of all edges in layer $0$, i.e., solved the original problem.

\subsection{Handling High-Degree Nodes}\label{ssec:high-degree-dp}

In \cref{ssec:high-degree} we replaced high-degree nodes with $O(1)$-depth subtrees; we will have both \emph{original} and \emph{auxiliary} edges in the tree. In general, this will result in a new DP problem, with possibly different rules for different edges. For our running example, MaxIS, the rules can be specified as follows:
\begin{itemize}
	\item Original edge $(u,v)$: if we have $u \in X$, we must have $v \notin X$, and vice versa.
	\item Auxiliary edge $(u,v)$: if we have $u \in X$, we must have $v \in X$, and vice versa.
\end{itemize}
\noindent In essence, this ensures that all new nodes that represent one original node make the same consistent choice. A similar strategy works for a wide range of graph problems.


\section{Further Applications}\label{sec:further}
In this section we discuss further applications of our framework. We start by discussing the challenge of processing high-degree nodes in problems that are not as simple as the MaxIS problem.

\subsection{The Tree Median Problem} \label{subsec:tree_median}

In the \emph{tree median} problem, the input is a rooted tree, where each leaf node has a number associated with it. The task is defined recursively: the label of a node has to be the median of the labels of its children. 

For nodes with even number of children, we require it to output the smaller of the two medians. This allows us to assume w.l.o.g. that all nodes have an odd number of children as those with even number of children can add a dummy leaf child with value $-\infty$.

This problem admits a simple sequential strategy, in which we label the nodes starting from the leaf nodes. However, as mentioned earlier, the tree median problem does not belong to the class of problems considered in the prior work by \cite{BateniBDHM18icalp,hajiaghayi18tree_dp}. In particular, tree median is not \emph{binary adaptable} as we cannot replace a high degree node with a binary tree (without using significantly larger dynamic programming tables and hence more memory). Nevertheless, in this section, we show that this problem can be solved by our algorithm in $O(\log D)$ rounds (in optimal space).

\subsubsection{Handling High Degree Nodes}
We replace the children of each high-degree node $u$ with an $O(1)$-diameter tree, as described in \cref{ssec:high-degree}. Recall that the original children of $u$ are leaves in this tree, and each internal node has degree at most $n^{\delta/2}$. Since we do not care about the output computed at these newly created nodes, we will call them \emph{don't-care} nodes. The don't-care nodes will hold some intermediate values that help $u$ compute the correct median. Throughout the process, we also remember the original parent of each node.

\subsubsection{Indegree-Zero Clusters}
An indegree-zero cluster at layer $i$ lies in a single machine. Therefore, the medians for all the nodes in the cluster can be locally computed from the medians computed for the clusters at layer at most $i-1$. If this cluster contains high degree nodes, we can compute the median for all such nodes in parallel by sorting all nodes by original parent identifier and median value, and picking the median of the children of these high degree nodes.

\subsubsection{Indegree-One Clusters}
An indegree-one cluster at layer $i$ consists of a unique directed path $P$ of $\ell$ nodes $p_{j}$ $(0 \le j < \ell)$, such that the incoming edge to the cluster is incoming to $p_{\ell-1}$ and the outgoing edge from the cluster is outgoing from $p_{0}$. For $j>0$, each node $p_{j} \in P$ can have an arbitrary number of incoming edges from other nodes and lower layer clusters one child $p_{j-1} \in P$. The output of $p_{j}$ is the variable $x_{j}$, initially unknown.

\begin{lemma}\label{lemma:path-reduce-leaves}
  For each node $p_{j-1}$, it is sufficient to store two values $a_{j}$, $b_{j}$ such that $x_{j-1} = \median(x_{j}, a_{j}, b_{j})$ and $a_{j} \le b_{j}$.
\end{lemma}
\begin{proof}
  If $p_{j-1}$ has more than two leaf children, it can delete all but the middle two and the median is still preserved. This can be done in parallel for all nodes in the tree by sorting by parent identifier and value.
\end{proof}

We will show how to replace a sub-path of length $2$ by a single edge. This process can be repeated in parallel by all nodes to compress the path $P$ in $O(1)$ rounds.

The situation is as follows: we have a path $p_{2} \to p_{1} \to p_{0}$, with $x_{0} =  \median(x_{1}, a_{1}, b_{1})$ and $x_{1} =  \median(x_{2}, a_{2}, b_{2})$. We wish to write $x_{0} =  \median(x_{2}, a, b)$. We can do it according to the following rule:
$$
(a, b) = \begin{cases}
  (a_{1}, a_{1}) & \text{if } b_{2} \le a_{1} \\
  (b_{1}, b_{1}) & \text{if } b_{1} \le a_{2} \\
  (\max\{a_{1}, a_{2}\}, \min\{b_{1}, b_{2}\}) & \text{otherwise}
\end{cases}
$$

\begin{lemma}
  Value $x_{0}$ is the correct median we wish to compute.
\end{lemma}
\begin{proof}
  Notice that $x_{0} = \median(x_{1}, a_{1}, b_{1}) \in [a_{1}, b_{1}]$, and $x_{1} = \median(x_{2}, a_{2}, b_{2}) \in [a_{2}, b_{2}]$. We now do a case analysis:

  If $b_{2}\le a_{1}$, it means that $x_{1} \le a_{1}$, and therefore $x_{0} = a_{1} = \median(x_{2}, a_{1}, a_{1})$.

  Similarly, if $b_{1}\le a_{2}$, it means that $x_{1} \ge b_{1}$, and therefore $x_{0} = b_{1} = \median(x_{2}, b_{1}, b_{1})$.

  The final case is that $[a_{1}, b_{1}]$ and $[a_{2},b_{2}]$ have an intersection which is the interval $[a, b]$, where
  \begin{align*}
  a &= \max\{a_{1}, a_{2}\} \\
  b &= \min\{b_{1}, b_{2}\}.
  \end{align*}
  It is easy to see that $x_{0} = x_{2}$ iff $x_{2} \in [a,b]$. And if $x_{2} \notin [a,b]$, then $x_{0} = a$ if $x_{2} < a$ and $x_{0} = b$ if $x_{2} >  b$. Hence $x_{0} = \median(x_{2}, a, b)$.
\end{proof}

If $P$ does not contain any (originally) high degree nodes, the entire cluster is locally processed by writing the value $x_{j}$ of all nodes $p_{j}\in P$ as $\median(x_{\ell}, a_{j}', b_{j}')$ where $a_{j}'$ and $b_{j}'$ are calculated by repeated application of the procedure mentioned for the length-2 path.

Now let $p_{j-1} \in P$ be a high degree node. We can write $p_{j-1} = \median(x_{j},a_{j},b_{j})$ just like in \Cref{lemma:path-reduce-leaves} by sorting by original parent identifiers. Here, $x_{j}$ is the value of the don't-care child $p_{j} \in P$, but we want $x_{j}$ to be the value of the original child of $p_{j-1}$ that belongs to $P$. In order to propagate the correct value, we have a different rule for the don't-care nodes in $P$ where they just copy the value of their child. Hence every high degree node in $P$ will also compute the correct median value.


\subsection{Bayesian Tree Inference and Belief Propagation}\label{subsec:inference}

Probabilistic or Bayesian graphical models are ubiquitous in machine learning and statistics \cite{Koller_book09}. A probabilistic graphical model is a graph, where the nodes (say, $x_i \in \mathbb{R}^{d_x}$) present hidden random variables with a conditional distribution structure defined by the vertices of the graph. We also get measurements of the graph (say, $y_i \in \mathbb{R}^{d_y}$) and an important problem of inference in graphical models is to compute the posterior distributions of the nodes, that is, $p(x_k \mid y_1,\ldots,n)$ for some selected $k = 1,\ldots,n$.

In this section, we consider an important special case of a Bayesian graphical model, where the graph is a tree and the observations are conditionally independent observations obtained at each node from a given conditional distribution model $p(y_i \mid x_i)$. The conditional distributions of the nodes then take the form $p(x_i \mid x_{\gamma_i})$, where $\gamma_i$ is the collection of child indices of the node $x_i$ (a leaf $j$ has $\gamma_j = \emptyset$). It now turns out that the algorithm framework presented in this paper allows us to compute $p(x_k \mid y_1,\ldots,n)$ in $O(\log D)$ MPC rounds, at least in the Gaussian special case which we consider here. We assume, without loss of generality, that $x_k$ corresponds to the root of the tree.

Let us denote the clique indices of the node $i$ as $\alpha_i = \{ i \} \cup \gamma_i$ and define clique potentials as
\[
  \psi_i(x_i,x_{\gamma_i}) = \psi_i(x_{\alpha_i}) = p(y_i \mid x_i) \, p(x_i \mid x_{\gamma_i}).
\]
The computation of posterior probability density $p(x_k \mid y_1,\ldots,n)$ then corresponds to computing the marginal of the product of the clique potentials:
\[
  p(x_k \mid y_1,\ldots,n) \propto \int \cdots \int \prod_{i=1}^n \psi_i(x_{\alpha_i}) \, d(x_{1:n \backslash k}).
\]
An efficient algorithm for solving this kind of problems on trees is called belief propagation \cite{Koller_book09}. In the case of path graphs (i.e., when each $\alpha_i$ is a pair of indices), the solution to the inference problem is given by Bayesian filters and smoothers \cite{Sarkka:2013}, and belief propagation corresponds to so-called two-filter smoother. Parallel algorithms for the Bayesian filtering and smoothing problems (i.e., inference for probabilistic path graphs) have recently been developed in \cite{Sarkka:2021,Hassan:2021}, but not in the context of the MPC model. However, the associative formulations used in those algorithms provide practical means for path compression that we also need in Bayesian trees.

If we now think that the present tree is actually the subtree within the current cluster, then we have the following two possible cases to consider:
\begin{enumerate}
\item {\em Indegree-zero cluster}, where we want to compute
\[
  \bar{\psi}_1(x_1) = \int \cdots \int \prod_{i=1}^n \psi_i(x_{\alpha_i}) \, d(x_{2:n}),
\]
where $x_1$ is the root. The potential $\bar{\psi}_1(x_1)$ then corresponds to compression of the indegree-zero cluster into a single node.

\item {\em Indegree-one cluster}, where we want to compute
\[
  \bar{\psi}_{j \to 1}(x_1, x_j) = \int \cdots \int \prod_{i=1}^n \psi_i(x_{\alpha_i}) \, d(x_{2:n \backslash j} )
\]
for some index $j \in \{ 2,\ldots,n \}$. Here $\bar{\psi}_{j \to 1}(x_1, x_j)$ corresponds to compression of the cluster into a node $x_1$ with an open child position $x_j$.
\end{enumerate}
For concreteness, let us now take a look at a linear Gaussian graph in which we have (for $i = 1,\ldots,n$):
\[
\begin{split}
  p(x_i \mid x_{\gamma_i}) &= \mathcal{N}(x_i; \sum_{j \in \gamma_i} F_j \, x_j + c_i, Q_i), \\
  p(y_i \mid x_i)  &= \mathcal{N}(y_i; H_i \, x_i + d_i, R_i), 
\end{split}
\]
that is,
\[
\begin{split}
  \psi_i(x_{\alpha_i}) = \mathcal{N}(y_i; H_i \, x_i + d_i, R_i) \, \mathcal{N}(x_i; \sum_{j \in \gamma_i} F_j \, x_j + c_i, Q_i),
\end{split}
\]
where $\mathcal{N}(x ; \mu,\Sigma)$ denotes a multivariate Gaussian probability density with mean vector $\mu$ and covariance matrix $\Sigma$. The representation of a node thus consists of the $|\gamma_i|$ matrices $\{ F_j~:~j \in \gamma_i\}$ along with $c_i$, $Q_i$, $y_i$, $H_i$, $d_i$, and $R_i$.

The implementation of the indegree-zero cluster operation (1) above requires just one primitive operation: The elimination of a leaf. We can repeat this operation until the whole tree is reduced into a single node. However, we need to ensure that we are able to do this operation in constant memory per node. Luckily, this is what happens in the Gaussian case.

Let us now consider a tree where we have added an additional node $x_{n+1}$ which is attached to the node $j$. What happens is that this adds a new child to node $j$:
\[
  \psi_j(x_{\alpha_j}) \to \tilde{\psi}_j(x_{\alpha_j}, x_{n+1}),
\]
and we also need to multiply with the leaf potential $\psi_{n+1}(x_{n+1})$. Thus, the joint potential is $\psi(x_{1:n+1}) =$
\[
\begin{split}
   \left[ \prod_{i=1}^{j-1} \psi_i(x_{\alpha_i}) \right]
   \, \tilde{\psi}_j(x_{\alpha_j}, x_{n+1}) \, \psi_{n+1}(x_{n+1})
   \left[ \prod_{k=j+1}^{n} \psi_k(x_{\alpha_k}) \right],
\end{split}
\label{eq:psi_np1}
\]
which we want to integrate over everything but $x_1$ in the present indegree-zero cluster case and over everything but $x_1,x_j$ in the indegree-one cluster case below. The elimination of the leaf in both cases corresponds to integration over $x_{n+1}$. 

The integration over $x_{n+1}$ can be done in closed form in the Gaussian case. In practice, it consists of computing the posterior covariance and mean parameters of $\psi_{n+1}(x_{n+1})$ which are 
\[
\begin{split}
  \tilde{Q}_{n+1} &= \left[ Q_{n+1}^{-1} + H_{n+1}^\top R_{n+1}^{-1} H_{n+1}   \right]^{-1}, \\
  \tilde{b}_{n+1} &= \tilde{Q}_{n+1} \left[ H_{n+1}^\top R_{n+1}^{-1} (y_{n+1} - d_{n+1}) + Q_{n+1}^{-1} c_{n+1} , \right],
\end{split}
\]
and then fusing them to the mean and covariance parameters of its parent node:
\[
\begin{split}
  c_j &\leftarrow F_{n+1} \, \tilde{b}_{n+1} + c_j, \\
  Q_j&\leftarrow  Q_j + F_{n+1} \tilde{Q}_{n+1} F_{n+1}^\top,
\end{split}
\]
which both are operations that can be done in constant memory.

For implementing the indegree-one cluster operation (2), we can first use the leaf elimination procedure above repeatedly to reduce the indegree-one cluster into a single indegree-one path. What we then have left is a path of the form (with reindexed intermediate nodes)
\[
  \psi_1(x_1,x_2) \, \psi_2(x_2,x_3) \, \psi_3(x_3,x_4) \, \times \cdots \times \, \psi_{j-1}(x_{j-1},x_j) 
\label{eq:path1}
\]
which we want to integrate over $x_{2:j-1}$. This can be implemented using pairwise combinations of the potentials, which can be done recursively as
\[
  \bar{\psi}_{m+1 \to 1}(x_1,x_{m+1}) = \int \bar{\psi}_{m \to 1}(x_1,x_{m}) \, \psi_{m}(x_{m},x_{m+1}) \, dx_m 
  \label{eq:psi_recursion}
\]
with initial condition $\bar{\psi}_{2 \to 1}(x_1,x_2) = \psi_1(x_1,x_2)$. This is a special case of the Kalman filter's associative rule derived in \cite{Sarkka:2021,Hassan:2021} (though backwards in time) and hence it can be implemented in constant additional memory for storing the temporary variables. The algorithm gives parameters $(A,b,C,\eta,J)$ which define a factorization of the form:
\[
\begin{split}
  \bar{\psi}_{j \to 1}(x_1,x_j) &\propto \mathcal{N}(x_1 ; A \, x_j + b, C) \, \mathcal{N}_I(x_j ; \eta, J) \\
  &= \mathcal{N}(x_1 ; A \, x_j + b, C) \, \mathcal{N}(x_j ; J^{-1} \, \eta, J^{-1}).
\end{split}
\]
The term $\mathcal{N}(x_j ; J^{-1} \, \eta, J^{-1})$ can now be fused to the measurement model at the node $j$ by finding an artificial measurement $z_j$ along with $\bar{H}_j$ and $\bar{R}_j$ such that $\mathcal{N}(x_j ; J^{-1} \, \eta, J^{-1}) \, N(y_j \mid H_j \, x_j + d_j, R_j) \propto N(z_j \mid \bar{H}_j \, x_j, \bar{R}_j)$. This can be done in constant memory by simple matrix and vector operations. In conclusion, the path compression just requires us to compute the parameters of the conditional distribution $\mathcal{N}(x_1 ; A \, x_j + b, C)$ and to form the artificial measurement model $N(z_j \mid \bar{H}_j \, x_j, \bar{R}_j)$ for the node $x_j$. This produces a new graph which we can continue to process recursively.

\subsection{Constructing Non-Standard Representations}\label{subsec:from_standard}
In \cref{sec:input_rep}, we saw how we can obtain the input in form of list-of-edges from various other representations in $O(\log D)$ rounds. In this section, having our algorithm in hand we will show that how
we can transform list-of-edges into other representations in $O(\log D)$ rounds. Let $A[(a_1,b_1),(a_2,b_2), \ldots (a_k,b_k)]$ be an array that contains a list-of-edges representation of a tree i.e. each index of $A$ contains a pair integer that represents a child node and its parent in tree.

\paragraph{List-of-Edges $\rightarrow$ Pointers-to-Parent:}

It is sufficient to sort $A$ by $a_i$, and then replace $(a_i,b_i)$ by $b_i$.

\paragraph{List-of-Edges $\rightarrow$ BFS-Traversal:}

We use our algorithm to compute the depth $d_i$ of each node $a_i$ in $O(\log D)$ rounds. Replace each entry $(a_i,b_i)$ with $(a_i,d_i)$
in the array $A$ now sort the array $A$ according to $d_i$ to obtain the BFS-traversal. The overall computation is done $O(\log D)$ in rounds.  

\paragraph{List-of-Edges $\rightarrow$ DFS-Traversal:}

First each node computes size of the subtree rooted at the node, which can be done in $O(\log D)$ rounds using our algorithm. Let $v_1,v_2, \ldots v_k$ be the children of a node $u$ such $t_i$ is size of the subtree rooted at $v_i$. Label the edge $(v_i,u)$ with the value $\sum_{j<i} t_j$ (note that $(v_1,u)$ gets label $0$). This is prefix-sum operation, which can be done $O(1)$ rounds in MPC. Let $l(u,v)$ denote the label of an edge $(u,v)$. Now we can compute the DFS-traversal time-stamp for each node $v$, denoted as $t(v)$ follows: set $t(v) = t(\text{parent}(w))+ l(v,\text{parent}(w))+ 1$. This is dynamic programming problem that we can solve in $O(\log D)$ rounds. Sorting the nodes according to their time-stamp will give us the DFS-traversal. 

\paragraph{DFS-Traversal $\rightarrow$ String-of-Parentheses:}

Given array $A$ representing the DFS-traversal, we find
the depth of each node in $O(\log D)$ rounds. Let $d_i$ be the depth of node $i$ in the tree. Each computer scans its part of the array from left to right and can
compute its part of the string as follows. Repeat the following steps for an index (or node in the tree) $i$ starting from $0$.
\begin{itemize}
    \item If $i$ is the only node in tree, add ``()'' in the string and exit.
    \item If $i$ is the root node, add ``('' in the string.
    \item If $d_{i+1}= d_i+1$, add ``('' in the string.
    \item If $d_{i}+k = d_i$, add $\underbrace{)) \ldots )}_\text{$k$-times}$ ( in the string.
    \item If $i$ is the last index in the array, add ``))'' in the string, one for the node and one for the root node.        
\end{itemize}

\section{Conclusions}\label{sec:conclusion}

In this work, we showed how a broad class of \emph{dynamic programming problems} can be solved in trees in the MPC model, with a relatively simple three-step approach:
turn the input into a standard representation in $O(\log D)$ rounds,
construct a hierarchical clustering in $O(\log D)$ rounds, and
solve the problem of interest in $O(1)$ rounds.
We expect that the hierarchical clustering will find applications also beyond the scope of dynamic programming problems.

One key open question is what happens once we step outside trees. The natural first step would be to consider bounded-treewidth graphs. Is it possible to find a similar hierarchical clustering efficiently also in bounded-treewidth graphs? And if so, does it still let us solve dynamic programming problems in constant time, given the hierarchical clustering?

\section*{Acknowledgment}
	We are grateful to Alkida Balliu, Darya Melnyk, and Dennis Olivetti for several fruitful discussions, and to the anonymous reviewers for their helpful feedback on prior versions of this work.
	This work was supported in part by the Academy of Finland, Grants 321901 (Gupta and Vahidi) and 334238 (Latypov and Pai).

\bibliographystyle{plainurl}
\bibliography{mybib}

\end{document}